\documentclass[conference]{IEEEtran}
\usepackage[T1]{fontenc}
\usepackage[utf8]{inputenc}
\usepackage{amsmath,amssymb,amsfonts,amsthm}
\usepackage{bbm}
\usepackage{mathrsfs}
\usepackage{xcolor}

\IEEEoverridecommandlockouts

\newtheorem{theorem}{Theorem}
\newtheorem{lemma}{Lemma}
\newtheorem{proposition}{Proposition}
\newtheorem{definition}{Definition}
\newtheorem{remark}{Remark}
\newtheorem{corollary}{Corollary}

\newcommand{\tr}{\mathrm{tr}}
\newcommand{\eins}{\mathbbm{1}}

\title{Deterministic Multi-User Identification over Bosonic Channels\\
\author{\IEEEauthorblockN{Gökhan Elmas, Janis N\"otzel (\emph{Member, IEEE})}
\textit{Emmy-Noether Group Theoretical Quantum Systems Design},\\
\textit{Technical University of Munich, Munich, Germany},\\
\textit{\{gokhan.elmas,janis.noetzel\}@tum.de}}
\thanks{This work was financed by the DFG via grant NO 1129/2-1 and by the Federal Ministry of Education and Research of Germany via grants 16KISQ093, 16KISQ039 and 16KISQ077. The generous support of the state of Bavaria via Munich Quantum Valley, the NeQuS- and the 6GQT project is greatly appreciated. Finally, the authors acknowledge the financial support by the Federal Ministry of Education and Research of Germany in the programme of “Souverän. Digital. Vernetzt.”. Joint project 6G-life, project identification number: 16KISK002. }}

\begin{document}

\maketitle

\begin{abstract}
We study deterministic multi-user identification over bosonic channels using coherent-state signatures. Each user is assigned a coherent product state under an average energy constraint, and identification is performed by a user-specific binary quantum test. In contrast to classical multi-user identification models based on shared codebooks, this formulation associates each receiver with a geometric signature in high-dimensional phase space.
Using metric entropy bounds, we show that the identification capacity exhibits a near-$k \log k$ scaling behavior.

\end{abstract}


\section{Introduction}

Identification via channels, introduced by Ahlswede and Dueck, differs fundamentally from ordinary transmission: the receiver is not required to decode the transmitted message, but only to decide whether a given message was sent \cite{AhlswedeDueck1989}. This new definition of communication task led to a substantially different scaling law -- namely, the number of messages that can be \emph{identified} within $k$ transmissions scales doubly exponential. In order to achieve this doubly exponential scaling, a randomized encoding scheme was utilized. When instead a deterministic strategy is pursued, the number of bits per channel use turned out to be of the order $k\log k$. Later work extended identification to feedback, strong converses, and joint transmission settings \cite{AhlswedeDueck1989Feedback,HanVerdu1992,hayashi}. More recently, identification has received great attention in both classical and quantum settings.

On the quantum side, Winter established quantum identification as a distinct information-theoretic task and extended it to settings with prior correlation and feedback \cite{Winter2004QCMIQC, Winter2006PriorCorrelationFeedback}, while Hayden and Winter connected it to weak decoupling \cite{HaydenWinter2012}. A second strand studies robust and secure identification for classical-quantum channels. Boche, Deppe, and Winter showed that for compound cq-channels the identification capacity matches the transmission capacity, while for wiretap cqq models a sharp dichotomy emerges between zero and positive secrecy capacity \cite{BocheDeppeWinter2019}. This made the cq setting both operationally meaningful and mathematically tractable.

The recent classical literature also developed deterministic identification, where no local randomness is available. Salariseddigh, Pereg, Boche, and Deppe characterized deterministic identification under power constraints, showing in particular that the Gaussian capacity is infinite \cite{SalariseddighPeregBocheDeppe2022}. This line of work was extended to Poisson, MIMO broadcast, and finite-output channels, revealing superlinear scaling laws \cite{SalariseddighJamaliPeregBocheDeppeSchober2023,RosenbergerPeregDeppe2023MIMO,Colomer2025Identification}. In parallel, Rosenberger, Deppe, and Pereg carried identification ideas into quantum broadcast channels \cite{RosenbergerDeppePereg2023QuantumBC}, while Colomer, Deppe, Boche, and Winter introduced a quantum hypothesis-testing lemma for deterministic identification over quantum channels \cite{Colomer2025Identification}. Modern converse techniques and Gaussian scaling laws are summarized in \cite{BruecheMrossZhaoLabidiDeppeJorswieck2024,colomerGaussian}. Moreover, under heterodyne detection the bosonic identification problem reduces to a deterministic identification problem over a classical Gaussian channel.

Bosonic channels are central to optical communications, where coherent states provide a canonical and experimentally relevant signaling family \cite{holevoBook,Weedbrook2012}. In this paper we study a deterministic multi-user identification model directly in terms of coherent-state geometry. Rather than designing a shared codebook and asking each receiver to decode or identify components of a transmitted tuple, we associate each receiver with a dedicated coherent-state signature in high-dimensional phase space. Identification is then performed via a binary test against that signature.

The key idea is geometric. The distinguishability of coherent-state signatures under displaced thermal noise decays exponentially in the squared Euclidean distance between their amplitude vectors. Consequently, assigning sufficiently separated signatures yields exponentially small false-alarm probabilities, reducing the problem to packing in $\mathbb{C}^k \cong \mathbb{R}^{2k}$, where squared distance controls the error exponent.

Our main contribution is an explicit deterministic construction together with a matching converse in the coherent-state signature model. On the achievability side, we combine displaced thermal-state bounds with metric entropy estimates for Euclidean balls to construct signatures with controlled pairwise separation. On the converse side, we show that any identification scheme with
uniformly controlled first- and second-kind errors induces a packing
of the same Euclidean ball, whose cardinality is upper bounded by a
covering number estimate. These two arguments together imply the order-optimal scaling law
\begin{align}
\log M_k = k\log k - k\log\log k + \mathcal O(k)
\end{align}
when the identification errors are required to vanish polynomially.

\section{Model}
We utilize the standard notation of quantum optics as in \cite{serafiniBook} and \cite{holevoBook}. The Fock space is denoted as $\mathcal F$, and displaced thermal states as $S_N(\alpha)$. The value $N\geq0$ is the number of noise photons and $\alpha\in\mathbb C$ the displacement value. In a communication setting, $\alpha$ models the received signal and $N$ additive noise. Positive-operator valued measurements (POVMs) $\{D_m\}_{m=1}^M$ model the actions allowed at the receiver side. They satisfy $D_m\geq0$ for all $m=1,\ldots,M$ and $\sum_mD_m=\eins$. To model our communication setting, we consider transmission over an idealized bosonic channel. In this setting, each user $m$ is assigned a physical signature modeled by a coherent product state 
\begin{align}
|\alpha_m^k\rangle
=
|\alpha_{m,1}\rangle \otimes \cdots \otimes |\alpha_{m,k}\rangle,
\end{align}
where $\alpha_m^k=(\alpha_{m,1},\ldots,\alpha_{m,k})\in\mathbb{C}^k$. The signature of each user is supposed to satisfy the average energy constraint
\begin{equation}
\label{eq:power-constraint}    
\|\alpha_m^k\|^2
=
\sum_{t=1}^k |\alpha_{m,t}|^2
\le k\cdot E,
\qquad
m\in[M_k],
\end{equation}
for some fixed $E>0$. To identify user $m$, the sender transmits the coherent state $|\alpha_m^k\rangle$. However, the receiver only receives its signature under some added level of thermal noise, so that instead of making a measurement on $|\alpha_m^k\rangle$ itself it must operate on the displaced thermal states $S_N^k(\alpha_m^k):=S_N(\alpha_{m,1})\otimes\ldots\otimes S_N(\alpha_{m,k})$ instead: Each user $m$ performs its individual binary quantum test $
\{\Pi_m,\, \eins-\Pi_m\},
$ where $0\leq\Pi_m\leq\eins$. The outcome associated with $\Pi_m$ is interpreted as ``yes, user $m$ was called''. As will become clear later, a capacity-achieving test for user $m$ is given by $\Pi_m=|\alpha^k_m\rangle\langle\alpha^k_m|$. This leads to the following definition of code:
\begin{definition}[Code]
An $(k,M_k,\lambda_{1,k},\lambda_{2,k})$ deterministic multi-user identification scheme consists of signatures $\{|\alpha_m^k\rangle\}_{m=1}^{M_k}$ satisfying $\|\alpha_m^k\|^2 \le k\cdot E$
and binary tests $\{\Pi_m, \eins-\Pi_m\}$ such that
\begin{align}
1-\tr [S_N^k(\alpha_m^k)\Pi_m ] &\le \lambda_{1,k},
\label{eq:first_kind}
\\
\tr[S_N^k(\alpha_{m'}^k)\Pi_m]&\le \lambda_{2,k},
\qquad m'\neq m.
\label{eq:second_kind}
\end{align}
\end{definition}
\begin{remark}
This is a multi-user identification model in which each receiver tests only its own identity. The model is deterministic because there is no encoder randomization.
\end{remark}
Due to our restriction to coherent states with energy constraint
$\|\alpha_m^k\|^2 \le k\cdot E$, it is natural to approach the construction
of a code via packings and coverings of the Euclidean ball of radius
$\sqrt{k\cdot E}$ in $\mathbb{C}^k$. A similar approach has already been successfully used in \cite{colomerGaussian}, and in our case the major difference is in the decoder, which is conceptually very simple and allows us to prove a coding theorem with matching achievability- and converse part in the first order.
\section{Result and Sketch of Proof}

\begin{theorem}[Explicit multi-user identification bound]\label{thm:main}
Fix $E>0$. For every $k\in\mathbb{N}$ and every $\rho_k>0$ satisfying
\begin{align}
2\rho_k \le \sqrt{k\cdot E},
\end{align}
there exists a code such that
\begin{align}
M_k
\ge
\left(\frac{\sqrt{k\cdot E}}{2\rho_k}\right)^{2k}
=
\left(\frac{k\cdot E}{4\rho_k^2}\right)^k,
\end{align}
with errors of first- and second kind given by
\begin{align}
\lambda_{1,k}\leq2^{-k\cdot\Lambda(\delta,N)},\qquad \lambda_{2,k}\le e^{-4\rho_k^2\cdot\Theta(\delta,N)}
\end{align}
where $\Lambda(\delta,N):=(N+\delta)\ln(\tfrac{N+\delta}{N})-(N+\delta+1)\ln(\tfrac{N+\delta+1}{N+1})$ and $\Theta(\delta,N):=\frac{1-(N+1)^{-1/(N+\delta)}}{N+1-N(N+1)^{-1/(N+\delta)}}$.
\end{theorem}
While Theorem \ref{thm:main} shows the existence of codes, the statement of a capacity theorem will in addition need a converse proof, which we provide below:

\begin{theorem}[Converse bound]\label{thm:converse}
Consider any deterministic multi-user identification scheme
in the coherent-state signature model satisfying
\begin{equation}
\|\alpha_m^k\|^2 \le k\cdot E, \qquad m \in [M_k],
\end{equation}
and suppose that the error probabilities satisfy
\begin{equation}
\max\{\lambda_{1,k}, \lambda_{2,k}\} \le \delta_k,
\qquad 0 < \delta_k < \tfrac{1}{2}.
\end{equation}
Then
\begin{equation}
M_k
\le
\left(
1 +
\frac{4\sqrt{k\cdot E}}{\sqrt{(2N+1)\log\!\bigl(1/(4\delta_k)\bigr)}}
\right)^{2k}.
\end{equation}
\end{theorem}
Theorems \ref{thm:main} and \ref{thm:converse} allow us to state the following:
\begin{corollary}
    For the noisy bosonic classical-quantum channel $\alpha\to S_N(\alpha)$ we have 
    \begin{align}
        C_{ID} = 1.
    \end{align}
\end{corollary}
\begin{remark}
    As can be seen from Theorem \ref{thm:main}, the separation parameter $\rho_k$ controls the tradeoff between code size and error probabilities: increasing $\rho_k$ improves the second-kind error exponent while decreasing the achievable number of users.
\end{remark}
While Theorem \ref{thm:converse} follows from a packing argument for Euclidean balls
combined with the Fuchs--van de Graaf inequality, the proof of Theorem \ref{thm:main}
relies on several key steps, which we outline below before providing details. First, we need an overlap formula for coherent states:
\begin{lemma} For all $\alpha^k, \beta^k \in \mathbb{C}^k$,
\begin{equation}
\operatorname{tr}\!\bigl(S_N^k(\alpha^k)\,|\beta^k\rangle\langle\beta^k|\bigr)
\le
\exp\!\left(-\frac{\|\alpha^k-\beta^k\|^2}{N+1}\right).
\end{equation}
\end{lemma}
\begin{IEEEproof}
By \cite{Marian_2007} and with $F$ denoting Fidelity,
\begin{align}
    \tr\bigl(S_N^k(\alpha^k)\,|\beta^k\rangle\langle\beta^k|\bigr) 
    &= \prod_{i=1}^k\tr\bigl(S_N(\alpha_i)\,|\beta_i\rangle\langle\beta_i|\bigr)\\
    &= \prod_{i=1}^kF(S_N(\alpha_i),|\beta_i\rangle\langle\beta_i|)\\
    &\leq \exp\!\left(-\frac{\|\alpha^k-\beta^k\|^2}{N+1}\right).
\end{align}
\end{IEEEproof}

The above identity immediately gives the error probabilities.
\begin{proposition}\label{prop:error-probabilities}
Let the photon-number threshold projector be $
P_{k,\delta}:=\sum_{n^k\in T_{N+\delta}^k} |n_1\rangle\langle n_1|\otimes\cdots\otimes|n_k\rangle\langle n_k|
$ where $T_{N+\delta}^k:=\{n^k:n_1+\ldots+n_k\leq k\cdot(N+\delta)\}$ and define for each user \(m\)
$
\Pi_m := D(\alpha_m^k)\,P_{k,\delta}\,D(\alpha_m^k)^\dagger$. Then the errors of first- and second kind satisfy
\begin{align}
\lambda_{1,k}&\le e^{-k\cdot \Lambda(\delta,N)},\\
\lambda_{2,k}&\le \exp\!\left( -\Theta(\delta,N)\cdot \min_{m\neq m'}\|\alpha_m^k-\alpha_{m'}^k\|^2 \right).
\end{align}
In particular, if the signatures satisfy the separation condition
\begin{equation}
\min_{m \neq m'} \|\alpha_m^k - \alpha_{m'}^k\| \ge 2\rho_k,
\end{equation}
then it follows that
\begin{equation}
\min_{m \neq m'} \|\alpha_m^k - \alpha_{m'}^k\|^2 \ge (2\rho_k)^2 = 4\rho_k^2,
\end{equation}
and therefore
\begin{equation}
\lambda_{2,k} \le \exp\bigl(-4\rho_k^2 \,\Theta(\delta,N)\bigr).
\end{equation}
\end{proposition}
Proposition \ref{prop:error-probabilities} again rests on tailbounds for the photon number distributions, which are the content of Lemma \ref{lem:tailbounds}, which we prove in Section \ref{sec:tailbounds}.
\begin{lemma}\label{lem:tailbounds}
    Let $\mathcal S_k(n^k):=n_1+\ldots+n_k$ for all $n^k\in\mathbb N^k$. For any $\gamma\in\mathbb C$, let $p$ be the probability distribution on $\mathbb N$ defined via $p(n|\gamma) := \langle n,S_N(\gamma)n\rangle$. For any $\gamma^k\in\mathbb C^k$, let $p^k(n^k|\gamma^k):=p(n_1|\gamma_1)\cdot\ldots\cdot p(n_k|\gamma_k)$ be a probability distribution on $\mathbb N^k$. Then
    \begin{align}
        \mathbb P(\mathcal S_k\geq k\cdot(N+\delta)) &\leq e^{-k\cdot\Lambda(\delta,N)}\label{eqn:tailbound-1}\\
        \mathbb P(\mathcal S_k\leq k\cdot(N+\delta))      &\leq e^{-k\cdot\|\gamma^k\|^2\cdot \Theta(\delta,N)},\label{eqn:tailbound-2}
    \end{align}
    where $\Lambda$ and $\Theta$ are defined in Theorem \ref{thm:main}.
\end{lemma}
\begin{IEEEproof}[Proof of Proposition \ref{prop:error-probabilities}]
Fix \(m\). If signature \(m\) is transmitted, then the received state is
\[
S_N^k(\alpha_m^k)=S_N(\alpha_{m,1})\otimes\ldots\otimes S_N(\alpha_{m,k}).
\]
By unitary invariance of the trace,
\begin{align}
\tr\!\bigl(S_N^k(\alpha_m^k)\Pi_m\bigr) &= \tr\!\bigl(S_N^k(0)\,P_{k,\delta}\bigr)\\
    &= \sum_{n^k\in\mathcal T_{N+\delta}^k}p(n_1|0)\cdot\ldots\cdot p(n_k|0)
\end{align}
hence $ \lambda_{1,k}
=
1-\operatorname{tr}\!\bigl(S_N^k(\alpha_m^k)\Pi_m\bigr)
=
1-\operatorname{tr}\!\bigl(S_N^k(0)\,P_{k,\delta}\bigr)
$.
The tailbound for correctly detecting the signature yields
\begin{equation}
\lambda_{1,k}\le e^{-k\Lambda(\delta,N)}.
\end{equation}

Now let \(m'\neq m\). Again by unitary invariance,
$
\operatorname{tr}\!\bigl(S_N^k(\alpha_{m'}^k)\Pi_m\bigr)
=
\operatorname{tr}\!\bigl(S_N^k(\alpha_{m'}^k-\alpha_m^k)\,P_{k,\delta}\bigr).
$
Set $
\Delta_{m,m'}^k:=\alpha_{m'}^k-\alpha_m^k.$
Applying the tailbound for misdetection of signatures with $\gamma^k:=\Delta_{m,m'}$  gives
\begin{equation}
\operatorname{tr}\!\bigl(S_N^k(\Delta_{m,m'}^k)\,P_{k,\delta}\bigr)
\le
\exp\!\left(-\Theta(\delta,N)\|\Delta_{m,m'}^k\|^2\right).
\end{equation}
Therefore
\begin{equation}
\lambda_{2,k}
\le
\exp\!\left(
-\Theta(\delta,N)\,
\min_{m\neq m'}\|\alpha_m^k-\alpha_{m'}^k\|^2
\right).
\end{equation}
If the minimum distance is at least \(2\rho_k\), then $
\lambda_{2,k}\le e^{-4\rho_k^2\Theta(\delta,N)}.$
This proves the claim.
\end{IEEEproof}

\section{Explicit Packing Bound via Metric Entropy}

We now derive an explicit lower bound on the number of users using metric entropy methods. Identify $\mathbb{C}^k$ with $\mathbb{R}^{2k}$, and let $B_{2k}(R) := \{x \in \mathbb{R}^{2k} : \|x\| \le R\}$ denote the closed Euclidean ball of radius $R$. Choosing signatures inside $B_{2k}(\sqrt{k\cdot E})$ ensures the energy constraint. For a metric space $(X,d)$ and $\varepsilon>0$, let $N(X,\varepsilon)$ denote the covering number, i.e., the smallest number of closed balls of radius $\varepsilon$ needed to cover $X$. Let $P(X,\varepsilon)$ denote the packing number, i.e., the largest cardinality of a subset with pairwise distances at least $\varepsilon$. Standard relations between packing and covering numbers imply
\begin{equation}\label{eq:pack}
P(X,2\varepsilon)\le N(X,\varepsilon)\le P(X,\varepsilon).
\end{equation}

We use the standard entropy estimate for Euclidean balls. In the present setting, it is enough to use the bounds
\begin{align}\label{eqn:basic-bounds}
\left(\frac{R}{\varepsilon}\right)^d
\leq
N(B_d(R),\varepsilon)
\leq
\left(1+\frac{2R}{\varepsilon}\right)^d,
\end{align}
which are classical and may be found, for example, in Szarek's discussion of metric entropy estimates for homogeneous spaces \cite{Szarek1998}. For completeness, we extract the consequence we need.
In particular, applying the lower bound in \eqref{eqn:basic-bounds} with $\varepsilon = 2\rho$ yields $N(B_d(R), 2\rho) \ge \left(\frac{R}{2\rho}\right)^d.$

\begin{lemma}
Let $d\in\mathbb{N}$, $R>0$, and $\rho>0$. Then there exists a set $
\{x_1,\ldots,x_M\}\subseteq B_d(R)
$
such that
\begin{equation}
\|x_i-x_j\| \ge 2\rho,
\qquad i\neq j,
\end{equation}
and
\begin{equation}
M \ge \left(\frac{R}{2\rho}\right)^d.
\end{equation}
\end{lemma}

\begin{IEEEproof}
Consider a maximal subset $\{x_1,\ldots,x_M\}\subseteq B_d(R)$ such that
\[
\|x_i-x_j\| \ge 2\rho, \qquad i\neq j.
\]
By maximality, the closed balls of radius $2\rho$ centered at the $x_i$ cover $B_d(R)$, and therefore $
N(B_d(R),2\rho)\le M.
$
On the other hand, the volumetric lower bound on covering numbers gives
\[
N(B_d(R),2\rho)\ge \left(\frac{R}{2\rho}\right)^d.
\]
Combining the two inequalities yields
\[
M \ge \left(\frac{R}{2\rho}\right)^d.
\]
This proves the claim.
\end{IEEEproof}

\begin{remark}
The bound is optimal up to constant factors in the radius term for Euclidean balls. In particular, for our purposes it captures the correct exponential dependence on the ambient dimension $d=2k$.
\end{remark}

Applying the lemma with $d=2k$ and $R=\sqrt{k\cdot E}$ yields the desired code-size estimates:
\begin{IEEEproof}[Proof of Theorem \ref{thm:main}]
By the packing lemma with $d = 2k$, $R = \sqrt{k\cdot E}$, and $\rho = \rho_k$,
there exists a set $\{\alpha_1^k,\dots,\alpha_{M_k}^k\} \subseteq \mathbb{C}^k \cong \mathbb{R}^{2k}$
such that
\[
\|\alpha_m^k\| \le \sqrt{k\cdot E}
\]
for every $m$, and
\begin{equation}
\|\alpha_m^k-\alpha_{m'}^k\| \ge 2\rho_k,\qquad m\neq m',
\end{equation}
with cardinality
\begin{equation}
M_k \ge \left(\frac{\sqrt{k\cdot E}}{2\rho_k}\right)^{2k}
=
\left(\frac{k\cdot E}{4\rho_k^2}\right)^k.
\end{equation}

Assign these vectors as user signatures. The energy constraint holds because
\begin{equation}
\|\alpha_m^k\|^2 \le k\cdot E.
\end{equation}

By Proposition~1,
\begin{equation}
\lambda_{1,k} \le e^{-k\Lambda(\delta,N)},
\end{equation}
\begin{equation}
\lambda_{2,k}
\le
\exp\!\left(
-\Theta(\delta,N)\min_{m\neq m'}\|\alpha_m^k-\alpha_{m'}^k\|^2
\right).
\end{equation}
Since the minimum pairwise distance is at least $2\rho_k$, we obtain
\begin{equation}
\lambda_{2,k}
\le
\exp\!\left(-\Theta(\delta,N)(2\rho_k)^2\right)
=
e^{-4\rho_k^2\Theta(\delta,N)}.
\end{equation}

This proves the theorem.
\end{IEEEproof}

\begin{IEEEproof}[Proof of Theorem \ref{thm:converse}]
Let $\rho_m := S_N^k(\alpha_m^k)$ and let $\{D_m\}_{m=1}^{M_k}$ be the corresponding tests. For every $m\neq m'$,
\begin{align}
\frac12\|\rho_m-\rho_{m'}\|_1
&\ge \tr(D_m\rho_m)-\tr(D_m\rho_{m'}) \\
&\ge 1-\lambda_{1,k}-\lambda_{2,k}.
\end{align}
By the Fuchs--van de Graaf inequality,
\begin{equation}
F(\rho_m,\rho_{m'})
\le 1-(1-\lambda_{1,k}-\lambda_{2,k})^2.
\end{equation}
For displaced thermal states,
\begin{equation}
F(\rho_m,\rho_{m'})
=\exp\!\left(-\frac{\|\alpha_m^k-\alpha_{m'}^k\|^2}{2N+1}\right),
\end{equation}
which implies
\begin{equation}
\|\alpha_m^k-\alpha_{m'}^k\|^2
\ge (2N+1)\log\!\frac{1}{1-(1-\lambda_{1,k}-\lambda_{2,k})^2}.
\end{equation}
Let 
\begin{equation}
r_k:=\min_{m\neq m'}\|\alpha_m^k-\alpha_{m'}^k\|.
\end{equation}
Then
\begin{equation}
r_k^2\ge (2N+1)\log\!\frac{1}{1-(1-\lambda_{1,k}-\lambda_{2,k})^2}.
\end{equation}
Since $\|\alpha_m^k\|^2 \le k\cdot E$, the signatures form an $r_k$-packing of
$B_{2k}(\sqrt{k\cdot E}) \subset \mathbb{R}^{2k}$.
By the packing–covering relation \eqref{eq:pack} and the upper bound in \eqref{eqn:basic-bounds},
\begin{align}
M_k 
&\le P(B_{2k}(\sqrt{k\cdot E}),r_k) \\
&\le N(B_{2k}(\sqrt{k\cdot E}),r_k/2) \\
&\le \left(1+\frac{4\sqrt{k\cdot E}}{r_k}\right)^{2k}.
\end{align}

and therefore
\begin{align}
M_k 
&\le \left(
1+\frac{4\sqrt{k.E}}{\sqrt{(2N+1)\log(1/(4\delta_k))}}
\right)^{2k}
\end{align}
\end{IEEEproof}
\section{Near-{$k\log k$} Scaling}

We now choose $\rho_k$ explicitly.
\begin{corollary}[Near-$k \log k$ achievability]
Fix $E > 0$ and $\gamma > 0$. Let $\rho_k^2 = \gamma \log k.$
Then for all sufficiently large $k$ there exists a deterministic
multi-user identification scheme satisfying
\begin{equation}
\lambda_{1,k} \le e^{-k\Lambda(\delta,N)},
\qquad
\lambda_{2,k} \le k^{-4\gamma\Theta(\delta,N)},
\end{equation}
and
\begin{equation}
M_k \ge \left(\frac{k\cdot E}{4\gamma \log k}\right)^k.
\end{equation}
Equivalently,$
\log M_k \ge k\log k - k\log\log k + k\log\!\left(\frac{E}{4\gamma}\right).
$
In particular,$
\log M_k = k\log k - k\log\log k + O(k).
$
\end{corollary}

\begin{IEEEproof}
Substituting
$
\rho_k^2 = \gamma \log k$
into the first-kind error bound of Theorem \ref{thm:main} gives
\begin{equation}
\lambda_{1,k} \le e^{-k\Lambda(\delta,N)}.
\end{equation}
Substituting the same choice into the second-kind error bound of
Theorem \ref{thm:main} gives
\begin{equation}
\lambda_{2,k}
\le
e^{-4\gamma\Theta(\delta,N)\log k}
=
k^{-4\gamma\Theta(\delta,N)}.
\end{equation}
Since $\rho_k^2 = \gamma \log k = o(k)$,
the condition $
2\rho_k \le \sqrt{k\cdot E}$
holds for all sufficiently large $k$.

Hence,
\begin{equation}
M_k \ge \left(\frac{k\cdot E}{4\gamma\log k}\right)^k.
\end{equation}
Taking logarithms gives
\begin{equation}
\log M_k
\ge
k\log k + k\log E - k\log(4\gamma) - k\log\log k.
\end{equation}
Rearranging,
\begin{equation}
\log M_k \ge k\log k - k\log\log k + k\log\!\left(\frac{E}{4\gamma}\right).
\end{equation}
The final $O(k)$ statement follows immediately.
\end{IEEEproof}

Combining the previous two corollaries gives the main scaling law.
\begin{theorem}[Order-optimal scaling law]
Fix $E>0$ and require
\[
\max\{\lambda_{1,k},\lambda_{2,k}\}\le \delta_k
\]
with $\delta_k$ vanishing polynomially in $k$. Then the maximal number of identifiable
users in the coherent-state signature model satisfies
\begin{equation}
\log M_k = k \log k - k \log \log k + O(k).
\end{equation}
\end{theorem}

\begin{IEEEproof}
The lower bound follows from Corollary~2, which constructs
schemes achieving $
\log M_k \ge k \log k - k \log \log k + O(k)$.
The upper bound follows from Theorem 2 applied to schemes with
$\max\{\lambda_{1,k},\lambda_{2,k}\}\le \delta_k$, where $\delta_k$
vanishes polynomially in $k$.
Combining the two bounds yields the claimed scaling law.
\end{IEEEproof}

\begin{remark}
The achieved scaling is not doubly exponential. However, it is still superlinear and differs from the $k\log k$ law only by the lower-order term $k\log\log k$.
\end{remark}

\section{Comparison to Heterodyne Detection}

We now compare the ideal coherent-state projection model with a physically realizable receiver, namely heterodyne detection.

Under heterodyne detection, a coherent-state input $|\alpha^k\rangle$ induces a classical output $Z^k \in \mathbb{C}^k$ with conditional distribution
\begin{equation}
P(z^k \mid \alpha^k) = \pi^{-k}
\exp\!\left(-\|z^k - \alpha^k\|^2\right)
\label{eq:heterodyne_density}
\end{equation}
This leads to a situation where equivalently, the receiver observes  Gaussian output with a minimum amount of so-called ``shot noise'' \cite{banaszekQuantumLimits}:
\begin{equation}
Z^k = \alpha^k + W^k,
\qquad
W^k \sim \mathcal{CN}(0, I_k).
\label{eq:heterodyne_awgn}
\end{equation}
so that heterodyne detection transforms the bosonic identification problem into a deterministic identification problem over a classical Gaussian channel with continuous output.

Consequently, any deterministic identification scheme based on heterodyne detection can be viewed as a classical identification scheme with input vectors $\alpha^k \in \mathbb{C}^k$ satisfying the energy constraint $\|\alpha^k\|^2 \le k\cdot E$, and observation governed by \eqref{eq:heterodyne_awgn}. The achievable performance and converse bounds are therefore determined by the corresponding results for Gaussian identification channels.

Such problems have recently been studied in \cite{colomerGaussian}, where deterministic identification over general linear Gaussian channels with continuous output is analyzed via rate--reliability tradeoffs. In particular, both achievable bounds and converse bounds can be expressed in terms of hypothesis-testing quantities for the induced Gaussian output distributions.


\section{Discussion}

The construction is geometric and deterministic. It avoids the shared-codebook coupling that often complicates multi-user identification problems by assigning one coherent-state signature per user. The detection rule is equally simple: each receiver performs a binary test against its own signature state. 

The resulting bounds are explicit and transparent. The key design parameter is the packing radius $\rho_k$, which directly controls both code size and error probability. The achievability result shows that one can support $
M_k \asymp \left(\frac{k}{\log k}\right)^k
$
forcing the identification errors to decay polynomially. The converse shows that this order cannot be improved within the present coherent-state signature framework.

The analysis is also physically interpretable. Coherent-state overlap is a natural bosonic distinguishability measure, and the energy constraint corresponds to restricting the signature vectors to a Euclidean ball. Thus the entire identification problem becomes a concrete high-dimensional geometry problem in phase space.

There are several natural directions for future work. First, one may replace the idealized projective test by more physically constrained measurements and quantify the resulting performance loss. Second, the present model can be generalized to noisy bosonic channels, where attenuation and thermal noise perturb the overlap geometry. Third, one may ask whether structured signature sets, rather than generic entropy-optimal packings, lead to improved constants or more practical implementations.

\section{CONCLUSION}

We presented a deterministic multi-user identification
framework for bosonic systems based on coherent-state sig-
natures. The central observation is that geometric separation
in phase space yields explicit control of the identification
errors under displaced thermal noise. Combining this with
metric entropy estimates for Euclidean balls yields explicit
lower and upper bounds on the number of users.

For every separation parameter $\rho_k$, we constructed
schemes with
\begin{equation}
M_k \ge \left(\frac{k\cdot E}{4\rho_k^2}\right)^k,
\ \
\lambda_{1,k} \le e^{-k\Lambda(\delta,N)},
\ \
\lambda_{2,k} \le e^{-4\rho_k^2\Theta(\delta,N)}.
\end{equation}
We also proved the matching converse upper bound
\begin{equation}
M_k \le
\left(
1+\frac{4\sqrt{k\cdot E}}{\sqrt{(2N+1)\log\!\bigl(1/(4\delta_k)\bigr)}}
\right)^{2k},
\end{equation}
showing that the achievable scaling is order-optimal in the
coherent-state signature model.

Choosing $\rho_k^2=\gamma\log k$ gives $
\lambda_{2,k}\le k^{-4\gamma\Theta(\delta,N)}
$
and
$
\log M_k = k\log k - k\log\log k + O(k).
$
This provides a clean and physically grounded near-$k\log k$
scaling law for deterministic multi-user identification in
bosonic channels.

\section{Tailbounds}
\label{sec:tailbounds}
\begin{IEEEproof}[Proof of Lemma \ref{lem:tailbounds}]
Since our code-words have energy bounded as $\|\alpha^k\|^2\leq k\cdot E$ and since the expected received energy if a receiver receives its signature is $S_N(0)$ and otherwise given by $N+|\alpha_i|^2$, the average received energy for a length-$k$ signal obeying the power constraint \eqref{eq:power-constraint} will be $k\cdot N$ for the correct code-word and will be lower bounded by $k\cdot(E+\epsilon)$ if the signal is intended to wake up a different receiver. The distributions $p(n|\alpha)$ are \cite[Eq. 2.10]{Marian_2007}) given by
\begin{align}
    p(n|\alpha) 
        &= \tfrac{1}{N+1}\left(\tfrac{N}{N+1}\right)^n e^{-\frac{|\alpha|^2}{N+1}}L_n\left(-\tfrac{|\alpha|^2}{N(N+1)}\right)
\end{align}
with $L_n$ being a Laguerre polynomial. In order to understand the large deviation behaviour of $p$ it is insightful to use its moment generating function which can be shown to equal
\begin{align}
    G(z,\alpha) &= \sum_{n=0}^\infty p(n)z^n\\
        &= \frac{1}{\pi N}\int e^{-|\gamma|^2/N}e^{-|\gamma+\alpha|^2(1-z)}d\gamma.
\end{align}
Using \cite[Eq. 5.51]{KamZhangFeng2023CoherentStates} we know the integral inequality
\begin{align}
    \frac{1}{\pi}\int_{\mathbb C} e^{-a\cdot |z|^2 + b\cdot z + c\cdot z^*}dz = \frac{1}{a}e^{b\cdot c/a}
\end{align}
we arrive at the equation 
\begin{align}
    G(z,\alpha) = \tfrac{1}{N+1-N\cdot z}\exp\big(-|\alpha|^2\tfrac{1-z}{N+1-N\cdot z}\big).
\end{align}
The moment generating function of $\mathcal S_k$ is given by
\begin{align}
    G_k(z,\alpha^k)=\mathbb E(S_k^z)=\tfrac{\exp\left(-\tfrac{\|\alpha^k\|^2\cdot(1-z)}{N+1-N\cdot z}\right)}{(N+1+N\cdot z)^k}.
\end{align}
\subsubsection{Correctly Detecting the Signature}
    From Markov's inequality, it follows that for every $s>0$
    \begin{align}
        \mathbb P(S_k\geq \xi)&=\mathbb P(e^{s\cdot S_k}\geq e^{s\cdot \xi})\\
            &\leq e^{-s\xi}\mathbb E(e^{s\cdot S_k})\\
            &= e^{-s\xi}G_k(e^{s}).
    \end{align}
    Hence, 
    \begin{align}
        \tr(P& [\otimes_{i}S_N(\alpha_i)]) = \mathbb P(S_k\geq k\cdot(N+\delta))\\
            &\leq e^{s\cdot k\cdot(N+\delta)}\frac{1}{(N+1+N\cdot e^{s})^k}\nonumber\\
            &\qquad\times\exp\big(-\|\alpha^k\|^2\cdot\frac{1-e^{s}}{N+1-N\cdot e^{s}}\big).
    \end{align}
    If $\|\alpha^k\|^2=0$ (true if receiver $m$ gets signal $m$) we have 
 
    \begin{align}
        \tr(P [S_N(0)^{\otimes k}])&\leq 
        e^{-k\cdot\|\alpha^k\|^2\cdot \Lambda(\delta,N)}.
    \end{align}    
    
\subsubsection{Misdetection of Signatures}
    Assume $\|\alpha^k\|_2^2\neq0$, (true if receiver $m$ receives code-word $m'\neq m$). Then,
    \begin{align}
        \tr((\eins-P)& [\otimes_{i}S_N(\alpha_i)]) = \mathbb P(S_k\leq k(N+\delta))\\
            &= \mathbb P(e^{-s\cdot S_k}\leq e^{-s\cdot k(N+\delta)})\\
            &\leq e^{s\cdot k(N+\delta)}(N+1-N\cdot e^{-s})^{-k}\nonumber\\
            &\qquad\times\exp\big(-\epsilon\cdot\frac{1-e^{-s}}{N+1-N\cdot e^{-s}}\big)\\
            &\leq e^{-k\cdot\|\alpha^k\|^2\cdot\Theta(\delta,N)}
    \end{align}
    can be shown by choosing $s=s^*$ such that $s^*  (N+\delta)  =\ln(N+1+N\cdot e^{-s^*})$.

\end{IEEEproof}

\bibliographystyle{plain}
\bibliography{bib}

\end{document}